%% file: vendor_competition.tex
\documentclass[a4paper,runningheads,envcountsect,envcountsame]{llncs}
\usepackage[numbers]{natbib}
\usepackage{times}
\usepackage{amssymb}
\usepackage{booktabs}
\usepackage{color,bbm}
\usepackage[colorlinks=true]{hyperref}
\usepackage{amsmath}
\usepackage{amsfonts}
\usepackage{multicol}

\def\R{\mathbb{R}}

\newcommand{\commentout}[1]{}

\newif\ifcomments
\commentstrue
\ifcomments
\newcommand{\colorthis}[2]{{[\color{#1}{#2}]}}

\else
\newcommand{\colorthis}[2]{}

\fi
\newcommand{\jo}[1]{\colorthis{blue}{Joel: #1}}

\commentout{
\newtheorem{theorem}{Theorem}
\newtheorem{corollary}[theorem]{Corollary}
\newtheorem{lemma}[theorem]{Lemma}
\newtheorem{proposition}[theorem]{Proposition}
\newtheorem{definition}[theorem]{Definition}
\newtheorem{claim}[theorem]{Claim}

}
\newtheorem{observation}[theorem]{Observation}
\newenvironment{proof-sketch}{\noindent{\bf Sketch of Proof}\hspace*{1em}}{\qed\bigskip}

\title{The Pricing War Continues: On Competitive Multi-Item Pricing}
\author{Omer Lev\inst{1}\and 
Joel Oren\inst{2}\and
Craig Boutilier\inst{2}\and
Jeffery S. Rosenschein\inst{1}
}
\begin{document}

\institute{
The Hebrew University, Israel\\
\email{omerl,jeff@cs.huji.ac.il}
\and
University of Toronto, Canada\\
\email{oren,cebly@cs.toronto.edu}}
\maketitle

\input{abstract.tex}
\input{intro.tex}
\input{prev_work.tex}
\input{preliminaries.tex}
\input{equivalence.tex}

\input{eq_analysis.tex}
\input{special_case.tex}
\input{conclusions.tex}

\bibliographystyle{splncs}
\bibliography{biblio}
\end{document}

%% file: abstract.tex
\begin{abstract}
We study a game with \emph{strategic} vendors who own multiple items and a single buyer with a submodular valuation function. The goal of the vendors is to maximize their revenue via pricing of the items, given that the buyer will buy the set of items that maximizes his net payoff.

 We show this game may not always have a pure Nash equilibrium, 
 in contrast to previous results for the special case where each
 vendor owns a single item. We do so by relating our game to an intermediate, discrete game
 in which the vendors only choose the available items, and their
 prices are set exogenously afterwards.

 We further make use of the intermediate game to provide tight bounds
 on the price of anarchy for the subset games that have pure Nash
 equilibria; we find that the optimal PoA reached in the previous special cases does not hold, but only a logarithmic one.

 Finally, we show that for a special case of submodular functions,
 efficient pure Nash equilibria always exist.
\end{abstract}

%% file: intro.tex
\section{Introduction}
Consider a scenario in the world of e-commerce, where a single
consumer is seeking to buy a set of products through an online website with
multiple vendors, such as Amazon or eBay. Given the items available
for sale and their prices, the buyer will purchase some subset
of them, according to his valuation of the items and their prices.

The vendors strive to
maximize their profits---for ease of exposition we assume
production costs are zero, hence profits can be equated with
revenues, or the sum of the prices of their
sold items.
A vendor can both competitively tailor the set
of items it offers and adjust
the prices of these items to react to their competitors.\footnote{Note 
that pricing an item sufficiently high can be
regarded as not offering it.} 
Indeed, automatic mechanisms for rapid online price optimization 
exist in a variety of markets and
industries~\cite{angwin2012coming}. This practice,
sometimes called \emph{competitive price intelligence}~\cite{skorupa14},
is a growing phenomenon within online retail. The specific
question that it addresses is how a company should price its products
in this competitive environment. 

Furthermore, as argued by Babaioff et al.~\cite{BabaioffNL14}, such a
setting introduces subtle algorithmic questions, since changing the
prices of the products may affect the resulting revenues in a
complex fashion, which may induce responses by one's
competitors. Therefore, studying the convergence properties of such
pricing dynamics is of interest.

In this paper, we take a game-theoretic approach to price competition
among multiple sellers, each with a set of multiple items. As in
Babaioff et al.~\cite{BabaioffNL14}, we study a setting with a single
buyer with a (combinatorial) valuation function, taken to be a
monotone and submodular set function over the set of items, which is fully known to the vendors. However, unlike that earlier
work, we examine the more general case where each of the $k$ vendors
controls a disjoint set of items $A_i$, rather than a single item. Given the prices of all of the
items, the buyer will buy the set with the highest net-payoff
(valuation minus the total price). Our model induces a game
in which each of the vendors' strategy is a pricing of their
items.

\paragraph{Contributions}
We begin by discussing a related two-phase game, that serves as a
way-station in our study of the main game. In this intermediate
game, vendors can only modify the sets of items being offered,
whereas their prices for these items are subsequently set by a specific pricing
mechanism. We show that this game, resulting only from this
modification of game dynamics (without changing its
parameters), has two critical properties. First, any
strategy profile in the original game (a price vector) has a
corresponding strategy profile that results in at least as much
revenue for each of the vendors (Proposition~\ref{prop:discrete}).

More importantly, we show that for any pure Nash equilibrium in the
original game, there is a corresponding pure Nash equilibrium in the
intermediate game, 
in which each vendor sells the same items at the same prices as
in the original equilibrium
(Theorem~\ref{thm:NE}). 
Hence, we reduce the pricing decision to a decision over what items to
sell, allowing a significant simplification of the problem.

We next study basic game theoretic properties of our game. We first show
that there are games which admit no pure Nash equilibrium. To do
so, we show that our two-phase game admits no pure Nash
equilibria, which then implies the nonexistence of a pure Nash
equilibrium in the original game, using Theorem~\ref{thm:NE}
(see Proposition~\ref{prop:no_PNE}).  This result suggests the
following question: suppose we restrict attention to instances of the game
that have some pure Nash equilibria---can we then say anything about
their value?  To accomplish this, we analyze the price of anarchy
(PoA) of this subset of games, where the objective function in
question is the social welfare value, taken to be simply the buyer's
valuation of the set of items that he purchased. We provide a tight
bound of $\Theta(\log m)$, where $m$ is the maximal number of items
controlled by any of the vendors.

Finally, as an additional way of dealing with the consequences of
Proposition~\ref {prop:no_PNE}, we investigate a special class of
valuation functions, which we call \emph{category-substitutable},
that, informally, partition products into ``equivalence
classes'' or categories, such that only a single item
will be chosen within a specific category, while 
different categories do not influence one
another. 
We show not only that \emph{efficient} pure Nash equilibria always
exist given such buyer valuations, but also provide a
precise characterization of such equilibria.

%% file: prev_work.tex
\section{Previous Work}

Multi-item pricing has been a significant topic  of research for many
years~\cite{HN12,HR12}, including analyses of the price of anarchy
(see \cite{CKS08} and follow-up papers).
The work of Babaioff et
al.~\cite{BabaioffNL14} is the most directly related to the model
developed here; indeed, the game that they study is a special case of
our game, in which each vendor sells only a single item. They show that for a
buyer with a general valuation function, a pure Nash equilibria may not
exist, though they prove several properties of the equilibria for games where 
one does. Furthermore, for \emph{submodular} valuation functions (the focus of
our paper), they show not only that pure Nash equilibria exist, but that they
are unique\footnote{More precisely, they give a closed-form characterization of
the prices of the items that \textbf{are sold}.} and efficient---i.e.,
have a \emph{price of anarchy (PoA)} of one. In contrast
to their setting, we show that in
our more general case there exist games with no pure Nash equilibria.
In cases where they do, we provide a characterization similar to theirs,
though in our more general case, the PoA is significantly higher.
 
Non-competitive (i.e., single-vendor) optimal-pricing problems
have been studied in the theoretical computer
science community. \cite{Guruswami:2005:PEP:1070432.1070598} study
a number of settings with \emph{multiple} buyers possessing various
valuation functions. They show that even with unit-demand buyers and an
unlimited supply of each item, selecting the optimal price vector is APX-hard;
they then provide a logarithmic approximation algorithm for the same case. It
should be noted that Babaioff et al.~also provide a $\log n$ approximation
algorithm for the case of a single vendor, and for that of a single buyer with
a submodular valuation function.


In a recent paper, Oren et al.~\cite{oren2014game} analyze a model in which
fixed prices are given exogenously, and there are multiple
unit-demand buyers. As above, their model assumes an unlimited supply
of each item.
The strategies of the vendors are which
\textbf{sets} of items to sell. Having the vendor make
decisions only about the set of items to sell has traditionally
been studied in the field of operation research. In particular, 
\emph{assortment optimization}~\cite{schon} deals with optimizing a seller's
``assortment'' (e.g., his catalog, or shelf), under various circumstances.
Although our game does not fall directly in this category, we do define a
discrete game in which the vendors' decisions are similar to
those in assortment
optimization (although the pricing procedure differs).


%% file: preliminaries.tex
\section{Preliminaries --- The Vendor Competition (VC) Game}

We consider the following setting: there is a set of $k$ vendors, with a corresponding vector of pairwise-disjoint sets of items $\mathbf{A}=(A_1,\ldots,A_k)$, such that $|A_i|=n_i$, and $n=\sum_{i=1}^{k}n_i$. We let $A^*=\bigcup_{i=1}^kA_i$.

A \emph{strategy profile} of the vendors is a price vector $\mathbf{p} \in
\mathbb{R}_{+}^n$, where $p(a)$ denotes the price of item $a$ according to
$\mathbf{p}$. For a set $S \subseteq A^*$, we let $p(S)=\sum_{a \in S}p(a)$.
For a vendor $i \in [k]$, we let $p_i \in \mathbb{R}_+^{n_i}$ denote vendor
$i$'s price vector for the items in $A_i$, and as before, for an item $a \in
A_i$, $p_i(a)$ denotes vendor $i$'s price for item $a$ according to $p_i$. For
convenience, we will let $\mathbf{p_{-i}}$ denote the price vector
corresponding to the items not in $A_i$ (of all other vendors).

\paragraph{The buyer's valuation function}

We assume that there is a single buyer with a valuation function $v:2^{A^*}
\rightarrow \R_{+}$; that is, the function $v(\cdot)$ assigns a non-negative
value to every \emph{bundle} (or subset) of items.  We let $m_a(S)=v(S \cup
\{a\}) - v(S)$, where $S \subseteq A^*$, denote
the \emph{marginal contribution} of
item $a$ to the set $S$. Following \cite{BabaioffNL14}, we assume that
$v(\cdot)$ is \emph{non-decreasing}: for $S \subseteq T \subseteq A^*$, $v(S) \leq
v(T)$ (implying that $v(\cdot)$ is maximized at $A^*$). Furthermore, we assume
that the valuation function is \emph{submodular}: for $S \subseteq T \subseteq A^*$
and $a \in A^* \setminus T$, we have that $m_a(S) \geq m_a(T)$. Both of these
assumptions are central in the model proposed by Babaioff et
al.~\cite{BabaioffNL14} (although additional discussion and results are provided
for non-submodular functions as well). Note that $v(\cdot)$ is said to be
submodular if the following, equivalent property holds: 
\[ v(S)+v(T) \geq v(S \cup T) + v(S \cap T),\quad \text{for all }S 
  \subseteq T \subseteq A^*. \]
Finally, slightly abusing notation, let the valuation be defined
over vectors of item sets as follows: for $S^* \subseteq A^*$, define
$\mathbf{S}=(S_1,\ldots,S_k)$ where $S_i = S^* \cap A_i$, for $i=1,\ldots,k$.
Then $v(\mathbf{S})=v(S^*)$. We adapt the rest of our function definitions in
an analogous fashion.

The buyer is assumed to have a \emph{quasi-linear} utility function: 
given a
vector of prices $\mathbf{p} \in \mathbb{R}_{+}^n$, the buyer's utility for a
bundle $S \subseteq A^*$ is $u_b(S,\mathbf{p}) = v_b(S) - \sum_{a \in
S}p(a)$. 
The \emph{demand correspondence} of the buyer is 
the family of sets that maximizes his
utility: 
$$D(v;\mathbf{p})=\{S \subseteq A^* : u_b(S) \geq u_b(S'), \forall  S' \subseteq A^* \}.$$ 
The buyer's
\emph{decision function}
$X(v;\mathbf{p})\subseteq 2^{A^{*}}$
 must satisfy
$X(v;\mathbf{p}) \in D(v;\mathbf{p})$. That is, given the price vector
$\mathbf{p}$, the buyer buys the bundle $X(v;\mathbf{p})$.  The buyer's
decision is said to be \emph{maximal} (or simply, the buyer is maximal), if
there does not exist a set $\tilde X \in D(v;\mathbf{p})$ such that
$X(v;\mathbf{p}) \subsetneq \tilde X$. Babaioff et al.~\cite{BabaioffNL14}
showed this property is critical to ensure the existence of a pure Nash
equilibrium in their setting. In our work, we will explicitly state where this
property is required.

\paragraph{Vendor payoffs}

Given the buyer's decision function $X$, and a (fixed) price vector
$\mathbf{p}=(p_{i},\mathbf{p_{-i}})$, vendor $i$'s utility is
$u_i^X(\mathbf{p})=\sum_{a \in X(v;\mathbf{p})\cap A_{i}}p(a)$. If the vendors
select mixed strategies, then a vendor's utility is defined to be his
\emph{expected} utility. 
Vendor $j$'s \emph{best response} to the other agents' mixed strategies is a
distribution over prices for $A_j$ that maximizes his expected utility.

This setup defines a game, parameterized by the vector $\mathbf{A}$ and the
valuation function $v$, in which each of the vendors prices his items to
maximize his utility. We will refer to such a game as a
\emph{vendor competition game}, or simply a 
\emph{VC game}.

When discussing our special case, in Section~\ref{semi-equivalent}, we will also make use of the following theorem, which was proved by Babaioff et al.:
\begin{theorem}[\cite{BabaioffNL14}]
\label{thm:babaioff14}
  Consider the case where each vendor owns a single item, that is $n_i = 1$ for every $i=1,\ldots,k$; that is, $A_i=\{a_i\}$, for $i=1,\ldots,k$. Then if the buyer's valuation function $v(\cdot)$ is non-decreasing and submodular, then there exists a pure Nash equilibrium, $\mathbf{p} \in \R_+^k$, of the following form: for every vendor $i$, such that $m_{a_i}(A^* \setminus a_i) > 0$, $p(a_i)=m_{a_i}(A^* \setminus a_i)$, and $a_i \in X(v;\mathbf{p})$. Also, the payoff of each vendor $i$ is precisely $m_{a_i}(A^* \setminus a_i)$.
\end{theorem}
\commentout{
\begin{theorem}[\cite{BabaioffNL14}]
\label{thm:babaioff14-approx}
Consider the case where each vendor owns a single item, that is $n_i = 1$ for every $i=1,\ldots,k$. Then if the buyer's valuation function $v(\cdot)$ is non-decreasing and submodular, and if there is only one vendor, then there is a $\Theta(\log n)$-approximate pricing algorithm. They also argue that this is tight, due to an inapproximation result given by Demaine et al. 
 \end{theorem}}
\paragraph{The objective function} Given a VC game $G=(v,\mathbf{S})$ and pricing vector $\mathbf{p} \in \R^n_+$, we use
the standard definition of social welfare, namely, the total
payoff of all the parties in the game, including the (non-strategic) buyer. Notice that by this definition, social welfare is simply the valuation of the set bought by the buyer, $v(X(v;\mathbf{p}))$, since all payments are simply transferred from the buyer to the vendors. We let $f(\mathbf{p})$ denote the social welfare resulting from a price vector $\mathbf{p}$.


%% file: equivalence.tex
\section{A Related Discrete Game}
\label{sec:discretization}

The game in its current formulation may seem somewhat hard to reason about, due
to large (continuous) strategy spaces.\footnote{In particular, the game is
clearly not in normal form. As such, we cannot directly apply Nash's theorem
about the existence of a mixed equilibrium. We defer the treatment of such
equilibria to future study.} To simplify our analysis, we use the
following discrete game, which can be thought of as imposing 
a specific pricing mechanism given the vendors' selection of
items (the design of the pricing scheme is influenced by
the results of~\cite{BabaioffNL14}).

\begin{definition}[The price-moderated VC game]
\label{def:pm-vc}
Given a buyer valuation function over the vendors' items, consider the 
following two-round process:
\begin{enumerate}
\item Each vendor $i \in [k]$ commits to offering
a subset of $S_i \subseteq A_i$ of items; this is its (discrete) strategy;
\item Given the strategy vector $\mathbf{S}=(S_1,\ldots,S_k)$, item prices are set to be their marginal values. That is, if we set $S^* = \bigcup_{i=1}^kS_i$, then for each $a \in S^*$, the mechanism will set $\tilde p (a)=m_a(S^* \setminus \{ a \})$. For each item $a' \notin S^*$ the mechanism sets $\tilde p(a') = v(A^*)+1$. Let $\mathbf{\tilde p}$ be the resulting price vector.
\end{enumerate}
The consumer then buys the set $X(v;\mathbf{\tilde p})$, as before. 
We call the resulting game a \emph{price-moderated vendor competition} game, or more succinctly, a \emph{PMVC} game.
\end{definition}
By analogy to our definitions for the original game, let $X'(v;\mathbf{S})$ denote the set of items sold, given the strategy profile $\mathbf{S}$. That is, given the price vector $\mathbf{\tilde p}$ imposed by the pricing mechanism in the second round, $X'(v;\mathbf{S}) = X(v;\mathbf{\tilde p})$. We similarly define a vendor's utility to be $u_i'(S_i,\mathbf{S_{-i}})$, for $i \in [k]$.

Note that the specified pricing ($v(A^{*}+1)$) of items not offered (i.e., not in $S^*$) ensures that the consumer will never buy them (i.e., $X(v;\mathbf{\tilde p}) \subseteq S^*$).
Further observe that the set of price vectors $\mathbf{\tilde p}$ that correspond to the discrete strategy profiles $\mathbf{S}$ in the PMVC game is a strict subset of the strategy space in the original VC game.
We justify our use of this game in our analysis by establishing the relationship between the original VC game and the proposed PMVC game, using a number of straightforward results.

\commentout{
Our following observation trivially follows from the fact that the induced strategy space of the PMVC game, resulting from translating the discrete strategy profiles to their corresponding price vectors as described above, can be viewed as subset of the strategy space of the original game.
\begin{observation}
Every resulting price vector $\mathbf{\tilde p}$ in the PMVC game, corresponding to a strategy profile $\mathbf{S}=(S_1,\ldots,S_k)$, can be translated to a strategy profile in the original VC game by setting the price of each sold item $a \in S^*$ to $m_{a}(S^*\setminus \{a\})$, and the price of every unsold item $a ' \notin S^*$ to be $v(A^*)+1$.
\end{observation}
}

\paragraph{Assumption} 

For ease of exposition, we assume that the buyer is maximal.
As we shall see, this implies that $X'(v;\mathbf{S})=X(v;\mathbf{p})$.
However, we can adapt the pricing mechanism by judiciously setting the prices
to be slightly below the marginal contributions to ensure maximality
(we leave the details of such a
modification to an expanded version of the paper).  We now
describe an important relationship between the VC and PMVC game that
simplifies our subsequent analysis by relating our original model to
to a simpler discrete game:
\begin{proposition}
\label{prop:discrete}
For every strategy profile $\mathbf{p}$ in the VC game and valuation
$v$, there is a strategy
profile $\mathbf{S}$ in the PMVC game such that $X'(v;\mathbf{S}) =
X(v;\mathbf{p})$, and $u'_i(S_i,\mathbf{S_{-i}}) \geq
u_i(p_i,\mathbf{p_{-i}})$ for each vendor $i$.
\end{proposition}
\begin{proof}
Let $\mathbf{p}$ be a strategy profile in the VC game, and let $T=X(v;\mathbf{p})$. Consider the strategy profile $\mathbf{S}$ where $S_i = X(v;\mathbf{p}) \cap A_i$, for $i=1,\ldots,k$, and let $\mathbf{\tilde p}$ be the resulting price vector imposed by the pricing mechanism. Furthermore, we let $\tilde T=X'(v;\mathbf{S})$. We begin by showing that $T=\tilde T$. First, notice that, as for all $a \notin T$, $\tilde p(a)=v(A^*)+1$, and hence item $a$ is not sold, and $\tilde T \subseteq T$. Next, suppose for the sake of contradiction that $\tilde T \subsetneq T$, and let $a \in T \setminus \tilde T$. By the submodularity of the function $v(\cdot)$, we have that $m_a(\tilde T) \geq m_a(T) \geq 0$. This implies that 
\begin{align*}
u_b(\tilde T \cup a, \mathbf{\tilde p})
&= v(\tilde T \cup a) - \sum_{a' \in \tilde T}\tilde p(a') - m_a(T\setminus a) \\
&\geq v(\tilde T \cup a) - \sum_{a' \in \tilde T}\tilde p(a') - m_a(\tilde T)\\
 &= u_b(\tilde T,\mathbf{\tilde p}).
\end{align*}
By maximality, the buyer would rather buy item $a$ as well, resulting in a contradiction.

We now claim that $u'_i(S_i,\mathbf{S_{-i}}) \geq u_i(p_i,\mathbf{p_{-i}})$. This follows from the fact that marginal contributions are the maximal prices at which the buyer still buys $X'(v;\mathbf{\tilde p})$. That is, any increase in the price would result in the buyer not buying the product: 
\begin{align*}
 u_b(\tilde T, \mathbf{\tilde p}) 
&= v(\tilde T) - \sum_{a' \in \tilde T \setminus a}\tilde p(a') - m_a( T \setminus a)\\
&= v(\tilde T \setminus a ) - \sum_{a' \in \tilde T \setminus a}\tilde p(a') = u_b(\tilde T \setminus a, \mathbf{\tilde p})
\end{align*}
\qed
\commentout{Every state in the general game induces a set of sold items $B\subseteq A$. 
We claim pricing all elements $b\in B$, $m_{b}(B)$, which is what happens in the discrete case, will only increase each players profit, so that the discrete game, where every player sells $B_{i}=B\cap A_{i}$ is a discrete game state ensuring each player at least as much utility as it got in its general game state. 
Obviously, every unsold item $a\in A\setminus B$ has $p_{a}>m_{a}(B\cup \{a\})$, otherwise it would have been sold (as it would have meant $v(B\cup \{a\})-p(B)-p_{a}>v(B)-p(B)$), contradicting the set $B$ being that of the sold items. However, if the price of items $b\in B$ were such that $p_{b}(B)<m_{b}(B)$, these prices could be increased to $m_{b}(B)$ and the items would still be sold (as $v(B)-p(B)-(m_{b}(B)-p_{b}(B))\geq v(B\setminus \{b\})-p(B\setminus \{b\})$), and the profit for the seller would increase.}
\end{proof}
Similarly to the previous proposition, which offered a mapping of strategy profiles in a way that does not cause the vendor's utilities to deteriorate, we now show that the same mapping also preserves Nash equilibria in cases where such equilibria exist. We note that the following result uses similar arguments to those given by Babaioff et al.~\cite{BabaioffNL14}, for proving a related characterization of pure Nash equilibria.
\begin{theorem}
\label{thm:NE}
For every pure Nash equilibrium $\mathbf{p}$ of a VC game there is a pure Nash equilibrium $\mathbf{S}=(S_1,\ldots,S_k)$ in the corresponding PMVC game, such that: (1) $X'(v;\mathbf{S})=X(v;\mathbf{p})$; and (2) for all $a \in X(v;\mathbf{p})$, $\tilde p(a)=p(a)$, where $\mathbf{\tilde p}$ is the induced price vector for $\mathbf{S}$. 
\end{theorem}
\begin{proof}
For convenience, let $B=X(v;\mathbf{p})$. As before, we let the strategy profile in the corresponding PMVC game be $\mathbf{S}=(S_1,\ldots,S_k)$, where $S_i = B \cap A_i$, for $i=1,\ldots,k$.

We begin by proving part (2) of the theorem. Suppose that there is an item $a\in B$ such that $p(a) \neq m_{a}(B \setminus a)$. If $p(a) > m_{a}(B) = v(B)-v(B\setminus a)$, then $v(B\setminus a)- p(B\setminus a)>v(B)-p(B)$, implying that the buyer would not buy item $a$, contradicting our assumption that $a\in B$.

Assume now that $p(a) < m_{a}(B)$. Letting $\mathbf{p'}$ denote the vector resulting by replacing $p(a)$ in $\mathbf{p}$ with $m_a(B \setminus a)$, we clearly have that $u_b(B, \mathbf{p'}) = u_b(B \setminus a, \mathbf{p'})$. We now prove the following claim:
\begin{claim}
 $u_b(B\setminus a,\mathbf{p'}) \geq u_b(T,\mathbf{p'})$, for all $T \subseteq B \setminus a$.
\end{claim}
\begin{proof}
Let $(B \setminus a ) \setminus T = \{c_1,\ldots,c_m\}$, and let $P_t=(B \setminus a) \setminus \{c_1,\ldots,c_t\}$ for all $t=1,\ldots,m$. Then $u_b(B \setminus a,\mathbf{p'}) \geq u_b(P_t,\mathbf{p'})$.
We prove the claim inductively. Suppose the claim is true for $t < m$: $u_b(B \setminus a,\mathbf{p'}) \geq u_b(P_t,\mathbf{p'})$. We now show that the same inequality holds for $t+1$ as well:
\begin{align*}
 u_b(P_{t+1},\mathbf{p'}) &= v(P_{t+1}) - p'(P_{t+1}) = v(P_{t}) - p'(P_t) - (m_{c_{t+1}}(P_{t+1}) - p(c_{t+1})) \\
&\leq u_b(P_{t},\mathbf{p'}) - (m_{c_{t+1}}(B) - p(c_{t+1})) \leq u_b(P_t, \mathbf{p'})
\end{align*}
where the first inequality follows from submodularity, and the second inequality follows from the fact that $p(c_{t+1}) \leq m_{c_{t+1}}(B)$ as we have previously shown.\qed
\end{proof}
Therefore, the vendor who owns $a$ can increase his payoff by setting the price of item $a$ to any value between $p(a)$ and $m_a(B)$, contradicting the equilibrium state.


What is left to prove is that $\mathbf{S}$ is a Nash equilibrium in the PMVC game. Note that we can assume w.l.o.g.\ the price of all products which are not sold is $v(A^*)+1$, as they remain unsold and continue to contribute nothing to the buyer or seller. Now, suppose $\mathbf{S}$ is not a Nash equilibrium, and that there is a player $i$, which can benefit from changing his set of sold items from $S_i$ to $S'_i$, which would result in a different vector of induced prices $\mathbf{\tilde p'}=(\tilde p'_i, \mathbf{\tilde p'_{-i}})$. We now argue that vendor $i$ can make an identical improvement in his revenue by changing his price vector from $p_i$ to $\tilde p'_i$, contradicting $p$ being a Nash equilibrium. For convenience, we let $B'=(B \setminus S_i) \cup S'_i$, and $B'' = X(v;\tilde p'_i,\mathbf{p_{-i}})$.

To show this, first notice no other vendor would sell any previously unsold items as a result; that is, $X(v;\tilde p'_i,\mathbf{p_{-i}}) \setminus A_i \subseteq X(v;p_i,\mathbf{p_{-i}}) \setminus A_i$ (since prices of items in $(A^*\setminus A_{i})\cap  X(v;p_i,\mathbf{p_{-i}})$ are still $v(A^*)+1$). So $B''=X(v;\tilde p'_i,\mathbf{p_{-i}}) \subseteq X'(v;S'_i,\mathbf{S'_{-i}})=B'$. Thanks to submodularity, we have that for every $a \in S'_i$, $m_{a}(B')<m_{a}(B'')$. Arguments similar to the ones given above (on $p(a)=m_{a}(B)$) imply that player $i$ would sell all the items in $S'_i$, and as the prices are unchanged from the PMVC game, will make the same profit as in the PMVC game. As this increases the player's profit in the PMVC game, it would increase its profit in the VC game as well, in contradiction to $p$ being a Nash equilibrium. \qed

\end{proof}

\paragraph{Discussion}
Note that we have not shown an exact equivalence between the two games: the set of Nash equilibria in the VC game is a subset of the equilibria in the PMVC game.
However, Proposition~\ref{prop:discrete} and Theorem~\ref{thm:NE} allow us to reason about our original game to a considerable extent. 

In contrast to the original model of Babioff et al.~in which $n_i=1$ for all $i=1,\ldots,k$, we can show that in our more general game, there may not always be a pure Nash equilibrium. In order to do so, we provide an example of a VC game in the next section with two vendors who each control two items. We show that this game does not admit any pure Nash equilibrium by relating to its corresponding PMVC game, using Theorem~\ref{thm:NE}.
Moreover, if we restrict ourselves to VC games that do admit pure Nash equilibria, we can provide quantitative bounds on their quality. Specifically, when restricting ourselves to VC games that have pure Nash equilibria, we provide a
lower bound on the price of stability of the PMVC game by analyzing an instance of the game. As the optimal objective value (the valuation of the set that is bought by the buyer) is always $v(N)$, Theorem~\ref{thm:NE} immediately implies that the same lower bound applies to the VC game. To complement our lower bound, we also provide an upper bound for the price of anarchy, also ensuring tightness of bounds.


%% file: eq_analysis.tex
\section{Equilibrium Analysis}

In Section~\ref{sec:discretization}, we outlined
several properties of the discrete PMVC game. We now describe
how the PMVC game can serve as a surrogate to help analyze
the stability of the VC game, and the quality of equilibria in the
VC game in cases when they exist.

\subsection{Existence of pure Nash equilibria}
We begin by showing that, as opposed to the special case where each vendor owns a single item, some instances of our game may not actually admit pure Nash equilibria. 
\begin{proposition}
  \label{prop:no_PNE}
There exists an instance of the VC game with two vendors, where $n_1=n_2=2$, that does not admit a pure Nash equilibrium.
\end{proposition}
\begin{proof}
  Let $A_1=\{a,b\}$ and $A_2=\{c,d\}$. We define the buyer's
  valuation function $v$ according to Table~\ref{tab:ce-valuation} (the
  value in each cell is the valuation of the union of the
  sets given at the head of the entry's row and column).
\commentout{
  $v(\emptyset ) = 0,
   v(\{b\})=2.503,
   v(\{d\})=2.703,
  v(\{c\})=2.803,
  v(\{a\})=3.203,
  v(\{c,d\})=4.1045,
  v(\{a,b\})=4.4045,
  v(\{b,d\})=5.204,
  v(\{a,d\})=v(\{b,c\})=5.304,
  v(\{a,c,d\})=v(\{a,b,d\})=6.5045,
  v(\{a,b,c\})=v(\{b,c,d\})=6.5045,$ and $v(\{a,b,c,d\})=7.6045$.
}

\begin{table}[ht]
\begin{minipage}[t]{.3\linewidth}
\begin{tabular}{@{}lllll@{}}
\toprule
            & $\emptyset$ & $\{c\}$ & $\{d\}$ & $\{c,d\}$ \\ \midrule
$\emptyset$ & 0           & 2.803   & 2.703   & 4.1045    \\
$\{a\}$     & 3.203       & 5.404   & 5.304   & 6.5045    \\
$\{b\}$     & 2.503       & 5.304   & 5.204   & 6.6045    \\
$\{a,b\}$   & 4.4045      & 6.6045  & 6.5045  & 7.6045    \\ \bottomrule
\end{tabular}
\caption{The buyer's valuation function}
\label{tab:ce-valuation}
\end{minipage}
\qquad\quad\quad
\begin{minipage}[t]{.3\linewidth}
\begin{tabular}{@{}lllll@{}}
\toprule
            & $\emptyset$ & $\{c\}$       & $\{d\}$       & $\{c,d\}$   \\ \midrule
$\emptyset$ & (0,0)       & (0,2.803)     & (0,2.7030)    & (0,2.703)   \\
$\{a\}$     & (3.203,0)   & (2.601,2.201) & (2.601,2.101) & (2.4,2.301) \\
$\{b\}$     & (2.503,0)   & (2.501,2.801) & (2.501,2.701) & (2.5,2.701) \\
$\{a,b\}$   & (3.103,0)   & (2.501,2.2)   & (2.501,2.1)   & (2.1,2.1)   \\ \bottomrule
\end{tabular}
\caption{The payoffs for each vendor.}
\label{ce-payoffs}
\end{minipage}
\end{table}
It is easy to verify that $v$ is (strictly) non-decreasing and
submodular. Now, consider the PMVC game with the same item sets and
valuation function $v$. For each strategy profile $(S_1,S_2)$,
the mechanism prices the items according to their marginal
contributions (Definition~\ref{def:pm-vc}). Therefore,
vendor payoffs are the sum of the prices of the items
offered. The vendors' payoffs for each strategy profile are given in
Table~\ref{ce-payoffs} (the first entry corresponds to the row
player, Vendor~1, and the second to the column player, Vendor~2). As is evident from
Table~\ref{ce-payoffs}, there is no pure Nash equilibrium in the PMVC
game. Theorem~\ref{thm:NE} then implies our proposition.
\end{proof}

\subsection{How bad can equilibria be?}

Given the negative nature of Proposition~\ref{prop:no_PNE},
we now 
restrict attention to the subclass of VC
games that \emph{do} admit pure Nash equilibria,
and ask whether reasonable
guarantees on social welfare in such equilibria can be derived

More formally, let $\mathcal{G}=\{G=(v,(A_1,\ldots,A_k)): \exists
\text{ a pure Nash equilibrium in G} \}$ be the set of VC
games which admit a pure Nash equilibrium. 
Define the \emph{price of anarchy (PoA)}
as follows:
\[ PoA_{\mathcal{G}} = \max_{G \in
  \mathcal{G}}\frac{\max_{\mathbf{p^*}}f(\mathbf{p^*})}{\min_{\mathbf{p}:\mathbf{p}\text{
    is a pure Nash equilibrium}}f(\mathbf{p})}\]
PoA is a commonly used worst-case measure of
the efficiency of the equilibria, and in our case reflects
the efficiency loss in
$\mathcal{G}$ resulting from the introduction of strategic
pricing, as opposed to using a ``centrally coordinated'' pricing policy.

\begin{theorem}
  Define the set of VC games $\mathcal{G}_{m}$, such
  that $G \in \mathcal{G}_m$ iff (1)
  $G$ has a pure Nash equilibrium, and
  (2) $max_{i=1}^k|A_i|=m$. Then the PoA of $\mathcal{G}_m$ is at most
  $H_m+1$, where $H_m$ is the $m$'th harmonic number.
\end{theorem}
\begin{proof}
  Consider a game $G=(v, \mathbf{A}=(A_1,\ldots,A_k)\}$ in  $\mathcal{G}_m$. It is enough to provide a lower bound on the minimal social
  welfare of a pure Nash equilibrium in the corresponding PMVC game: by
  Theorem~\ref{thm:NE}, this will establish a lower bound on the
  social welfare of a pure Nash equilibrium in $G$ as well. So let
  $\mathbf{S}=(S_1,\ldots,S_k)$ be a pure Nash equilibrium of the
  PMVC game. As $v(\cdot)$ is non-decreasing, we can assume without
  loss of generality that $|S_i|=\{a_i\}$, for some $a_i \in A_i$.

  Again by the assumption that $v(\cdot)$ is non-decreasing, we know
  that the optimal social welfare is obtained when all of $A^*$ is
  sold, so it is enough to upper bound $v(\mathbf{A})$ in terms of $v(\mathbf{S})$.

  We now show the following straightforward bound on the social
  welfare resulting from switching from $S_i$ to $A_i$:
  \begin{lemma}
\label{lem:poa_bound1}
    $v(A_i,\mathbf{S_{-i}}) \leq v(\emptyset,\mathbf{S_{-i}}) + H_{n_i} (v(S_i,\mathbf{S_{-i}}) -
    v(\emptyset,\mathbf{S_{-i}}))$, for all $i=1,\ldots,k$.
  \end{lemma}
  \begin{proof}
  As $\mathbf{S}$ is a Nash equilibrium, the profit from selling $a_{i}$ is higher than selling any set $B\subseteq A_{i}$. Using the definition of the pricing mechanism of the PMVC game, we know
  \[ v(S_i,\mathbf{S_{-i}}) - v(\emptyset, \mathbf{S_{-i}}) \geq
  \sum_{b \in B}m_b(B \setminus b,\mathbf{S_{-i}}), \quad\quad
  \text {for all } B \subseteq A_i \]
  
  By an averaging argument, this means that for all $B \subseteq A_i$, there exists an item $b \in
B$ such that
\begin{align}
\label{eq:bound1}
\frac{1}{|B|} (v(S_i,\mathbf{S_{-i}}) - v(\emptyset,
\mathbf{S_{-i}})) \geq m_b(B \setminus b,\mathbf{S_{-i}})
\end{align}
The above implies that there is a relabelling of the items in $A_i$,
so that: (1) $A_i=\{b_1,\ldots,b_{n_i}\}$, (2) $b_1 = a_1$, and (3) if
set $P_t = \{b_1,\ldots,b_t\}\cup \mathbf{S_{-i}}$ and $P_0=\mathbf{S_{-i}}$, the following holds:
\begin{align*}
  v(A_i,\mathbf{S_{-i}}) &= v(\emptyset,\mathbf{S_{-i}})+\sum_{t=1}^{n_i}m_{b_t}(P_{t-1}) \leq
  v(\emptyset,\mathbf{S_{-i}})+\sum_{i=1}^{n_i}\frac{1}{t}(v(S_i,\mathbf{S_{-i}}) - v(\emptyset,\mathbf{S_{-i}}))\\
& = v(\emptyset,\mathbf{S_{-i}})+H_{n_i}(v(S_i,\mathbf{S_{-i}}) - v(\emptyset,\mathbf{S_{-i}}))
\end{align*}
where the first equality follows from a simple telescopic series, and
the first inequality follows from Eq.~\ref{eq:bound1}.\qed
  \end{proof} 
Next, we show the following useful bound:
\begin{lemma}
\label{lem:poa_bound2}
  $\sum_{i=1}^k v(A_i,\mathbf{S_{-i}}) \geq v(A_i, \mathbf{A_{-i}}) +
  (k-1)v(S_i, \mathbf{S_{-i}}) $
\end{lemma}
\begin{proof}
  $L^{(t)}=(A_1,\ldots,A_t,S_{t+1},\ldots,S_k)$, for $t=1,\ldots,k$,
  and $L^{(0})=\mathbf{S}$. That is, $L^{(t)}$ is the strategy profile
  resulting from replacing the length-$t$ prefix of $\mathbf{S}$ with
  that of $\mathbf{A}$. 
  We prove by induction that
  \[\sum_{i=1}^t v(A_i,\mathbf{S_{-i}}) \geq v(L^{(t)}) +
  (t-1)v(\mathbf{S})\]
and the lemma would follow by setting $t=k$.
  
  The inequality clearly holds for $t=1$, due to the monotonicity of
  $v(\cdot)$. Assume  that the inequality holds for $t<k$. Thus, for $t+1$,
  we have:
  \[ \sum_{i=1}^{t+1} v(A_i,\mathbf{S_{-i}}) \geq v(L^{(t)}) +
  (t-1)v(\mathbf{S}) + v(A_{t+1},\mathbf{S_{-(t+1)}})\]
By the second definition of submodularity, we know $v(L^{(t)}) +
v(A_{t+1},\mathbf{S_{-(t+1)}}) \geq v(L^{(t+1)}) +v(\mathbf{S})$. Putting this in the preceding inequality concludes the proof.\qed
\end{proof}
We can also prove an upper bound on the optimal social welfare in terms
of the social welfare of $\mathbf{S}$. By the above two lemmas, we
get:
\begin{align*}
  v(\mathbf{A}) &\leq \sum_{i=1}^k v(A_i,\mathbf{S_{-i}})- (k-1)v(\mathbf{S})\\
  &\leq \sum_{i=1}^kv(\emptyset,\mathbf{S_{-i}}) +
  \sum_{i=1}^k H_{n_i}(v(S_i,\mathbf{S_{-i}}) -
  v(\emptyset,\mathbf{S_{-i}})) - (k-1)v(\mathbf{S})\\
& \leq \sum_{i=1}^kv(S_i,\mathbf{S_{-i}}) +
  \sum_{i=1}^k H_{n_i}(v(S_i,\mathbf{S_{-i}}) -
  v(\emptyset,\mathbf{S_{-i}})) - (k-1)v(\mathbf{S})\\ 
&= v(\mathbf{S})
  +\sum_{i=1}^k H_{n_i}(v(S_i,\mathbf{S_{-i}}) - v(\emptyset,\mathbf{S_{-i}})),
\end{align*}
where the third inequality follows from the monotonicity of
$v(\cdot)$.

By submodularity and that $n_i\leq m$ for $i=1,\ldots,k$, 
\begin{align*}
  v(\mathbf{A}) &\leq v(\mathbf{S}) + H_{m}\sum_{i=1}^{k}(v(S_i,\mathbf{S_{-i}}) - v(\emptyset,\mathbf{S_{-i}})) \\
&\leq v(\mathbf{S}) + H_m v(\mathbf{S})= v(\mathbf{S})(H_m +1),
\end{align*}
which establishes our upper bound on the PoA. 
\end{proof}
We also give an example of a game with
a pure Nash equilibrium that matches the above bound.
\begin{theorem}
\label{prop:POA-lb}
There exists a game in $\mathcal{G}_{m}$ with a price of anarchy
  of $H_m$.
\end{theorem}
\begin{proof}
Our counter-example is obtained by making the bound of
Lemma~\ref{lem:poa_bound1} tight. Consider a game $G =
(v,\mathbf{A}=(A_1,\ldots,A_k))$, in which $|A_i|=m$, for
$i=1,\ldots,k$.

We define the valuation function as follows. For a strategy
profile $\mathbf{T}=(T_1,\ldots,T_k)$, we set
$v(\mathbf{T})=\sum_{i=1}^k\ell(T_i)$, where $\ell(T_i)=0$ if
$|T_i|=0$, and otherwise we set $\ell(T_i)=H_{|T_i|}$. Observe that the vendors are all symmetric, and that
furthermore, the payoffs only depend on their own prices.

We now consider the following strategy profile $\mathbf{p}$.
Pick an arbitrary item
$a_i$ from each $A_i$, for $i=1,\ldots,k$, and set $p(a_i)=1$. Price the
remaining items at $v(A^*)+1$. Note that the payoff
of each vendor is precisely $1$.\footnote{If the buyer is not
  maximal, we can decrease the prices of the $a_i$'s by some small $\epsilon$.}

It is easy to see that $\mathbf{p}$ is a pure Nash
equilibrium. Indeed, suppose that it is not, and let $i$ be an
arbitrary vendor. Then he has an
alternative pricing $p'_i \neq p_i$, such that deviating to it would
improve his payoff of $1$. Suppose that the set of items being bought under a deviation
to $p'_i$ is $B=X(v;p'_i,\mathbf{p'_{-i}})$, such that $\sum_{a \in B}p(a)>1$. Then
there exists an item $b \in B$, such that $p(b) > 1/|B|$. But then by the
definition of the valuation function we have:
\begin{align*}
  \ell(B) - p(B) = H_{|B|} - p(B \setminus b) - p(b) <
  H_{|B|-1} - p(B \setminus b)
\end{align*}
contradicting the assumption that the buyer ends up buying set $B$.
\qed
\end{proof}
Note that the above construction can be extended to show that the
Price of Stability (PoS) is identical:
\begin{corollary}
  The Price of Stability of the class of games $\mathcal{G}_m$ is $\Omega(H_m)$.
\end{corollary}
\begin{proof-sketch}
  Use the construction in the proof for Proposition~\ref{prop:POA-lb},
  but set $\ell(T_i)=1$ if $|T_i|=1$, and
  $\ell(T_i)=H_{|T_i|}-\epsilon$, if $|T_i|>1$, for sufficiently small
  $\epsilon$. It is not hard to show that the aforementioned pure Nash
  equilibrium is the \emph{only} Nash equilibrium. A similar bound follows.
\end{proof-sketch}

%% file: special_case.tex
\section{Special Case: Product Categories}
\label{semi-equivalent}

A particular VC gane of interest is one in which have classes of items that are
roughly equivalent; as such the buyer is interested
in at most one item from each class (e.g., TV sets of a
certain size, with different manufacturers and sets of feature). 
Items is different classes however are ``unrelated'' so
the buyer's valuation for any set of items is additive \emph{across} 
these classes.
This scenario reflects the common case of shops selling very similar products, of
which the buyer only needs one, and tries to understand how the model's pricing behaviour.

\begin{definition}
A Category-Divided Substitutable-Product Vendor Competition game (CDSP-VC) is a VC game with a buyer that has a category-product-substitutable valuation function: 
\begin{itemize}
\item $A^*$ is partitioned into $r$ pairwise-disjoint sets, $T^{(1)},\ldots, T^{(r)}$. That is, $T^{(i)}\cap T^{(j)}=\emptyset$ for $i\neq j$, and $\bigcup_{i=1}^{r}T^{(i)}=A^*$. We refer to each set $T^{(j)}$ for $j=1,\ldots,r$, as a category.
\item For $S\subseteq A^*$, $v(S)=\sum_{i=1}^{r}v(S\cap T^{(i)})$.
\item For $S\subseteq A^*$ and $1\leq i\leq r$, $v(S\cap T^{(i)})=max_{a\in S\cap T^{(i)}}v(a)$.
\end{itemize}
\end{definition}

Due to the additivity of different item classes, we can focus on the pricing dynamic within a specific category and easily generalize the results.
\begin{observation}
For a category $T^{(j)}$, regardless of the other vendors' strategies, no vendor can profit by selling any items other than his most valuable one in category $T^{(j)}$.
\end{observation}
\begin{proof}
As the buyer only buys a single item from each category, vendors can make positive revenue from at most one item.
Consider a vendor $i$ and category $T^{(j)}$ such that $A_i \cap T^{(j)} \neq \emptyset$. 
Suppose that $a\in \arg\max_{b\in T^{(j)}\cap A_{i}}v(b)$, and consider a strategy profile $\mathbf{p}$. For any item $b \in (A_{i}\setminus a) \cap T^{(j)}$, if $b \in X(v;p_i,\mathbf{p_{-i}})$, then if vendor $i$ switches to a price vector $p'_i$, which prices $a$ the same as $b$ and prices $b$ at $v(A^*)+1$, then surely $a \in X(v;p'_i,\mathbf{p_{-i}})$, and the player's profit does not decrease.


\end{proof}
The above observation implies that within every category, each of the vendors is better off effectively trying to sell his highest valued item. In other words, we can assume without loss of generality, that for every vendor $i$ and category $T^{(j)}$, the vendor can pick an item $a^{(j)}_{i} \in \arg\max_{a \in T^{(j)} \cap A_i}v(a)$ (if such item exists) and set $p(b)=v(A^*)+1$, for all $b \neq a^{(j)}_i$, without incurring a loss as a result. Therefore, this reduces our game to $r$ independent special cases of the VC game, in which each vendor owns a \emph{single} item.

We turn to the result given by Babaioff et al.~\cite{BabaioffNL14} (Theorem~\ref{thm:babaioff14} in the preliminaries). Their result implies
the following characterization of the prices in a pure Nash equilibrium. 

\begin{corollary}
Every CDSP-VC game has a pure Nash equilibrium of the following form. For every category $T^{(j)}$, let $c_{i}^{(j)} = \arg\max_{a \in (T^{(j)}\cap A_{i})}v(a)$, and $w=\arg\max_{i} c_{i}^{(j)}$. Let $b^{(j)}=\arg\max_{a \in (T^{(j)}\setminus A_{w})}v(a)$ or $b^{(j)}=0$ if $|T^{(j)}|=1$. Then $p(c_{w}^{(j)})=v(c_{w}^{(j)}) - b^{(j)}$, and for every player $i\neq w$, $p(c_{i}^{(j)})=0$. For all other items $a\in T^{(j)}$, $p(a)=v(A^*)+1$.
\end{corollary}
\begin{proof}
Once we know that each player sells only a single product, products that are not sold need to be priced high, and the rest of the result stems from Theorem~\ref{thm:babaioff14}, as the above difference constitutes the marginal contribution of the item $c^{(j)}$.
\end{proof}


%% file: conclusions.tex
\section{Conclusions and Future Work}

In analyzing the multi-item, multi-vendor problem, we began by defining a discrete game, which allowed us to consider a related case on the road to analyzing the scenario of a ``multiple item per vendor'' game. The main property of the discrete game was to transform player strategies from pricing, to selecting what items to sell. To paraphrase Clausewitz's famous dictum, displaying (what to sell) became pricing by other means.

Utilizing this discrete game, we were able to prove that a multi-item, multi-vendor game with submodular buyer valuations does not necessarily have a Nash equilibrium (unlike the ``single item per vendor'' model), and furthermore, even when an equilibrium exists, the equilibrium may provide only a logarithmic price of anarchy (again, unlike the single item per vendor model). Building on these results, we showed that in a particular (yet, in our view, common) category-substitute model, there will always be an efficient pure Nash equilibrium.

There are many open problems in this area, even before the ``holy grail'' of pricing multi-item multi-buyer scenarios. We believe that there is a need to establish the characteristics of valuation functions that guarantee the existence of a Nash equilibrium. This is true both in our scenario, where even when assuming submodular functions a Nash equilibrium is not assured, as well as in simpler scenarios, such as the one item per vendor model, where, while equilibrium is assured for submodular valuations, it is not known if that is a characterization of the properties required for an equilibrium.

Adding more buyer valuations changes the model significantly, as vendors do not simply construct some ``buyer in expectation'' and act according to it, but rather have a wider range of options to pursue (primarily bundling). Perhaps using a metric to define a set of similar, yet not identical, buyers, it might be possible to build on our results (and those of others), and construct extensions to the current model incorporating multiple buyers.